\documentclass[sigconf,nonacm]{acmart}

\setcopyright{none}
\renewcommand\footnotetextcopyrightpermission[1]{}
\acmConference{}{}{}
\acmDOI{}
\acmISBN{}

\usepackage{amsmath}

\usepackage{amssymb}
\usepackage{amsthm}
\usepackage{booktabs}
\usepackage{algorithm}
\usepackage{algpseudocode}
\usepackage{tikz}
\usepackage{url}
\hypersetup{hidelinks}
\usetikzlibrary{shapes,arrows,positioning,calc,decorations.pathreplacing,fit,backgrounds}

\newtheorem{definition}{Definition}
\newtheorem{theorem}{Theorem}

\newtheorem{criterion}{Criterion}

\usepackage{microtype}

\begin{document}

\title{Witnessd: Proof-of-Process via Adversarial Collapse}

\author{David Condrey}
\affiliation{%
  \institution{Writerslogic Inc}
  \city{San Diego}
  \state{CA}
  \country{USA}
}

\begin{abstract}
Digital signatures prove key possession, not that the signer created the content. Detection systems classify content as authentic or inauthentic, but presume a ground truth oracle that adversarial settings lack. We propose a different criterion: evidence systems should be evaluated by how effectively they \emph{collapse the space of plausible alternative explanations}.

We introduce the \emph{jitter seal}, a proof-of-process primitive that creates cryptographic evidence of physical keyboard interaction. The system injects imperceptible microsecond delays (500--3000\,$\mu$s) derived via HMAC from a session secret, keystroke ordinal, and document state. Valid evidence could only exist if real keystrokes produced the document through those intermediate states.

We present Witnessd, an architecture combining jitter seals with Verifiable Delay Functions, external timestamp anchors (Bitcoin, RFC~3161), dual-source keystroke validation, and optional TPM attestation. Each component provides an independent verification checkpoint; disputing authentic evidence requires alleging coordinated compromise across multiple trust boundaries.

The system does not prevent forgery---a kernel-level adversary can defeat it. But such an adversary must make \emph{specific allegations}: naming a mechanism, identifying a time window, and implying testable claims. We formalize this as the \emph{Adversarial Collapse Principle} and evaluate across 31,000 jitter seal trials (100\% valid verification, 0\% forgery success) and 700+ system tests.
\end{abstract}


\maketitle

\section{Introduction}
\label{sec:intro}

Machine-generated evidence is increasingly contested. When a party claims ``this document existed at time $T$'' or ``this content was produced through process $P$,'' how should such claims be evaluated?

Current approaches fall into three categories:

\textbf{Detection} systems classify content as authentic or inauthentic. But detection presumes a ground truth oracle that adversarial settings lack. A ``95\% accurate'' detector is meaningless when the adversary controls whether to be in the 5\%.

\textbf{Certification} relies on trusted third parties. But trust is social, not technical. Certifiers can be wrong, compromised, or coerced.

\textbf{Cryptographic integrity} attests that bits were not modified after signing. But bit integrity is not process integrity. ``Signed by key $K$'' does not establish ``authored by human $H$.''

We propose a different framing: evidence systems should be evaluated not by detection accuracy or certification authority, but by how effectively they \emph{collapse the space of plausible alternative explanations}.

\subsection{The Gap Between Ownership and Process}

Alice signs a document with her private key. The signature proves she possesses the key and attests to the content. It does not prove she wrote it. She could have generated text with a language model, constructed intermediate states post-hoc, and signed each hash. A valid signature chain over document hashes $H_1, H_2, \ldots, H_n$ demonstrates commitment to each state---not that physical interaction produced those states.

This gap affects all signature-based provenance. As AI-generated text becomes indistinguishable from human writing, signature chains prove nothing about whether a human produced the content.

\subsection{Proof-of-Process as Primitive}

\begin{definition}[Proof-of-Process]
A \emph{proof-of-process} is cryptographic evidence that a physical process---not just a signing key---produced a digital artifact.
\end{definition}

Where proof-of-work demonstrates computational resources were expended, proof-of-process demonstrates physical interaction occurred. Table~\ref{tab:primitives} positions this primitive.

\begin{table}[ht]
\centering
\small
\begin{tabular}{@{}ll@{}}
\toprule
\textbf{Primitive} & \textbf{What It Proves} \\
\midrule
Digital signature & Key possession + content attestation \\
Proof-of-work & Computational resources expended \\
Verifiable delay function & Sequential computation time elapsed \\
\textbf{Proof-of-process} & \textbf{Physical interaction produced artifact} \\
\bottomrule
\end{tabular}
\caption{Cryptographic primitive categories.}
\label{tab:primitives}
\end{table}

\subsection{Contributions}

This paper contributes:

\begin{enumerate}
    \item The \textbf{Adversarial Collapse Principle} as an evaluation criterion for evidence systems (\S\ref{sec:adversarial})
    \item The \textbf{jitter seal}, a proof-of-process primitive with formal security arguments (\S\ref{sec:jitter})
    \item The \textbf{Witnessd architecture} combining verification checkpoints with independent trust assumptions (\S\ref{sec:architecture})
    \item \textbf{Security analysis} of residual attacks and required allegations (\S\ref{sec:security})
    \item \textbf{Evaluation} from implementations with 31,000+ verification trials (\S\ref{sec:eval})
\end{enumerate}

\subsection{Scope and Non-Goals}

The jitter seal proves typing occurred. It does \emph{not} prove:
\begin{itemize}
    \item The typist originated the ideas
    \item The typist did not copy from a source
    \item The typist did not type AI-generated content
\end{itemize}

These are fundamental limits of any process-based proof. A human who types AI output produces valid evidence. The system attests to physical process, not cognitive origin.

\section{Threat Model}
\label{sec:threat}

\subsection{Capability Stratification}

We stratify adversaries by capability level:

\begin{table}[ht]
\centering
\small
\begin{tabular}{@{}lll@{}}
\toprule
\textbf{Level} & \textbf{Capability} & \textbf{Example} \\
\midrule
$C_1$ & User-space access & Scripts, automation \\
$C_2$ & Administrator & Clock manipulation, logs \\
$C_3$ & Kernel & Driver injection, memory \\
$C_4$ & Infrastructure & External anchors, TSAs \\
\bottomrule
\end{tabular}
\caption{Capability stratification.}
\label{tab:capabilities}
\end{table}

Each evidence component forces allegations at a specific level; the architecture's value lies in requiring \emph{different} levels for different components.

\subsection{Adversary Classes}

\begin{table}[ht]
\centering
\small
\begin{tabular}{@{}lll@{}}
\toprule
\textbf{Adversary} & \textbf{Goal} & \textbf{Capability} \\
\midrule
Hostile verifier & Discredit evidence & Post-hoc examination \\
Hostile environment & Produce false evidence & Admin access, clocks \\
Post-hoc disputant & Create doubt & Vague allegations \\
\bottomrule
\end{tabular}
\caption{Adversary classes.}
\label{tab:adversaries}
\end{table}

\textbf{What we do not assume.} We do not assume trusted hardware, a trusted operating system, or trusted third parties. We assume only that standard cryptographic primitives are secure and that public external anchors are not simultaneously compromised.

\textbf{The goal.} We do not claim to prevent attacks. We claim to \emph{force specific allegations}: challenging the evidence requires naming a concrete mechanism, identifying a bounded time window, and making independently testable claims.

\section{The Adversarial Collapse Principle}
\label{sec:adversarial}

\subsection{Specific Allegation Criterion}

\begin{criterion}[Specific Allegation]
An allegation is \emph{specific} if and only if it:
\begin{enumerate}
    \item Names a concrete mechanism
    \item Identifies a bounded time window
    \item Implies a capability level
    \item Is independently testable
\end{enumerate}
\end{criterion}

``The document could have been backdated'' is not specific. ``The clock was set back between 14:00--14:47, requiring $C_2$ access to host X'' is specific.

\subsection{Adversarial Collapse}

\begin{criterion}[Adversarial Collapse]
Evidence produces \emph{adversarial collapse} when any alternative explanation requires a conjunction of specific allegations against components with \emph{independent} trust assumptions.
\end{criterion}

\begin{definition}[Trust Independence]
Two components $A$ and $B$ have \emph{independent trust assumptions} if:
\begin{enumerate}
    \item No single party controls both verification paths
    \item No shared secret enables compromise of both
    \item No common vulnerability affects both
    \item Compromise of $A$ does not enable compromise of $B$
\end{enumerate}
\end{definition}

\textbf{Where independence holds:} Local hash chain integrity (verified by recomputation) is independent of Bitcoin anchor validity (verified against public blockchain). TPM attestation (hardware root of trust) is independent of RFC~3161 TSA receipts (third-party PKI).

\subsection{Why Single Primitives Fail}

\begin{table}[ht]
\centering
\small
\begin{tabular}{@{}lll@{}}
\toprule
\textbf{Primitive} & \textbf{Why Insufficient} & \textbf{Vague Doubt} \\
\midrule
Keystroke timing & Replay, synthesis & ``Could be synthesized'' \\
Crypto integrity & Bit $\neq$ process & ``Proves bits, not origin'' \\
Certification & Trust assumption & ``Certifier wrong'' \\
Detection & No oracle & ``False negative'' \\
\bottomrule
\end{tabular}
\caption{Why single primitives fail adversarial collapse.}
\label{tab:primitives-fail}
\end{table}

Each primitive alone permits vague doubt. In combination, disputing any one requires specific allegations; disputing all requires allegations against \emph{independent} components.

\section{The Jitter Seal}
\label{sec:jitter}

The jitter seal is the proof-of-process layer. It creates cryptographic evidence that real-time keyboard interaction occurred during document creation.\footnote{Implementation: \url{https://github.com/writerslogic/physjitter}}

\subsection{Protocol Specification}

\begin{algorithm}
\caption{$\mathsf{Setup}(1^\lambda) \rightarrow (S, \mathsf{params})$}
\begin{algorithmic}[1]
\State $S \gets_\$ \{0,1\}^{256}$ \Comment{Session secret}
\State $\mathsf{params} \gets (N, J_{\min}, J_{\max})$ \Comment{Sampling interval, jitter range}
\State \Return $(S, \mathsf{params})$
\end{algorithmic}
\end{algorithm}

\begin{algorithm}
\caption{$\mathsf{Sample}(S, i, H_i, t_i, Z_i, B_i, J_{i-1}) \rightarrow (J_i, \sigma_i)$}
\begin{algorithmic}[1]
\State $\mathsf{mac} \gets \mathsf{HMAC\text{-}SHA256}(S)$
\State $\mathsf{mac}.\mathsf{update}(i)$ \Comment{8-byte ordinal}
\State $\mathsf{mac}.\mathsf{update}(H_i)$ \Comment{32-byte document hash}
\State $\mathsf{mac}.\mathsf{update}(t_i)$ \Comment{8-byte timestamp}
\State $\mathsf{mac}.\mathsf{update}(Z_i)$ \Comment{1-byte zone transition}
\State $\mathsf{mac}.\mathsf{update}(B_i)$ \Comment{1-byte interval bucket}
\State $\mathsf{mac}.\mathsf{update}(J_{i-1})$ \Comment{4-byte previous jitter}
\State $\mathsf{raw} \gets \mathsf{u32}(\mathsf{mac}.\mathsf{finalize}()[0..4])$
\State $J_i \gets J_{\min} + (\mathsf{raw} \mod (J_{\max} - J_{\min}))$
\State $\sigma_i \gets \mathsf{SHA256}(\mathsf{prefix} \| i \| t_i \| H_i \| J_i \| \sigma_{i-1})$
\State \Return $(J_i, \sigma_i)$
\end{algorithmic}
\end{algorithm}

\begin{algorithm}
\caption{$\mathsf{Verify}(S, \mathsf{evidence}) \rightarrow \{\mathsf{accept}, \mathsf{reject}\}$}
\begin{algorithmic}[1]
\For{each sample $(i, t_i, H_i, Z_i, B_i, J_i, \sigma_i)$}
    \State $(J'_i, \_) \gets \mathsf{Sample}(S, i, H_i, t_i, Z_i, B_i, J_{i-1})$
    \If{$J_i \neq J'_i$}
        \State \Return $\mathsf{reject}$
    \EndIf
    \If{$\sigma_i \neq \mathsf{SHA256}(\ldots)$}
        \State \Return $\mathsf{reject}$
    \EndIf
\EndFor
\State \Return $\mathsf{accept}$
\end{algorithmic}
\end{algorithm}

\subsection{Cryptographic Binding}

Each jitter value is bound to six components:

\begin{table}[ht]
\centering
\small
\begin{tabular}{@{}ll@{}}
\toprule
\textbf{Input} & \textbf{Binding Property} \\
\midrule
Session secret $S$ & Unforgeability \\
Keystroke ordinal $i$ & Ordering \\
Document hash $H_i$ & Document-binding \\
Zone transition $Z_i$ & Physical input pattern \\
Interval bucket $B_i$ & Typing rhythm \\
Previous jitter $J_{i-1}$ & Chain integrity \\
\bottomrule
\end{tabular}
\caption{Cryptographic binding of jitter seal components.}
\label{tab:binding}
\end{table}

The zone transition $Z_i$ encodes which keyboard region typed consecutive characters. We partition the keyboard into 8 zones based on touch-typing finger assignments. Zone transitions capture typing behavior without revealing content---each zone contains 3--6 keys.

The interval bucket $B_i$ quantizes inter-keystroke intervals into 10 bins:
$$B_i = \min\left(\left\lfloor \frac{t_i - t_{i-1}}{50\mathsf{ms}} \right\rfloor, 9\right)$$

\subsection{Security Properties}

\begin{definition}[Unforgeability]
The jitter seal is \emph{unforgeable} if, for all PPT adversaries $\mathcal{A}$:
$$\Pr[\mathsf{Verify}(S, \mathcal{A}(\mathsf{params}, \{H_i\})) = \mathsf{accept}] \leq \mathsf{negl}(\lambda)$$
where $\mathcal{A}$ has access to document hashes but not session secret $S$.
\end{definition}

\begin{definition}[Document-Binding]
Evidence for document $D$ does not verify for document $D' \neq D$:
$$\Pr[\mathsf{Verify}(S, \mathsf{evidence}_D, D') = \mathsf{accept} \land D \neq D'] \leq \mathsf{negl}(\lambda)$$
\end{definition}

\begin{theorem}[Unforgeability]
For any PPT adversary $\mathcal{A}$ making at most $q$ queries:
$$\Pr[\mathcal{A} \text{ wins}] \leq \frac{q}{2^{256}} + \left(\frac{1}{R}\right)^n$$
where $R = J_{\max} - J_{\min}$ is the jitter range.
\end{theorem}

\begin{proof}
Without the session secret $S$, the adversary must either recover $S$ from HMAC outputs (probability $\leq q/2^{256}$ by HMAC-PRF security~\cite{bellare1996}) or guess valid jitter values (probability $(1/R)^n$ for $n$ samples). With $R = 2500$ and $n = 100$: $(1/2500)^{100} < 2^{-1000}$.
\end{proof}

\subsection{What Jitter Seals Prove}

\begin{table}[ht]
\centering
\small
\begin{tabular}{@{}ll@{}}
\toprule
\textbf{Captured} & \textbf{Not Captured} \\
\midrule
Keystroke count & Which specific key \\
Timestamp & Character produced \\
Document hash & Document content \\
Jitter value & Cursor position \\
Zone transition & Application context \\
\bottomrule
\end{tabular}
\caption{Privacy-preserving evidence collection.}
\label{tab:privacy}
\end{table}

At no point does any component have simultaneous access to keystroke identity and persistent storage.

\section{Witnessd Architecture}
\label{sec:architecture}

The Witnessd architecture combines multiple verification checkpoints with independent trust assumptions.

\subsection{Verification Checkpoints}

\begin{table}[ht]
\centering
\small
\begin{tabular}{@{}lll@{}}
\toprule
\textbf{Layer} & \textbf{Property} & \textbf{Disputing Requires} \\
\midrule
0: Jitter seal & Process proof & Secret theft or typing \\
1: VDF proofs & Timing lower bound & VDF forgery \\
2: External anchors & Time upper bound & Third-party compromise \\
3: Dual-source & Hardware origin & Kernel compromise \\
4: TPM binding & Platform attestation & Hardware tampering \\
5: Hash chain & Append-only & Collision or DB access \\
\bottomrule
\end{tabular}
\caption{Verification checkpoints and required allegations.}
\label{tab:checkpoints}
\end{table}

\subsection{Layer 1: Verifiable Delay Functions}

VDFs prove minimum elapsed time between checkpoints. Given input $x$, a VDF computes output $y$ such that computing $y$ requires sequential work proportional to parameter $T$, while verification is efficient.

We use the Pietrzak VDF with $O(\log T)$ verification:

\begin{table}[ht]
\centering
\small
\begin{tabular}{@{}lll@{}}
\toprule
\textbf{Parameter} & \textbf{Value} & \textbf{Meaning} \\
\midrule
Modulus $N$ & RSA-2048 & Trusted setup \\
Time parameter $T$ & $2^{20}$ & $\sim$1M squarings \\
Proof size & $\sim$5 KB & 20 intermediate values \\
Verification ops & 40 & vs.\ 1M for recomputation \\
\bottomrule
\end{tabular}
\caption{VDF parameters.}
\label{tab:vdf}
\end{table}

The RSA-2048 modulus requires trusted setup; production deployments should use class group VDFs~\cite{wesolowski2019} which eliminate this requirement. Our implementation uses RSA for compatibility with existing libraries.

\subsection{Layer 2: External Anchors}

External anchoring provides third-party timestamp verification:

\textbf{OpenTimestamps (Bitcoin):} Root hash submitted to calendar servers, aggregated into Merkle tree~\cite{merkle1980}, embedded in Bitcoin transaction. After confirmation, inclusion proof demonstrates timestamp.

\textbf{RFC 3161 TSA:} TimeStampReq with SHA-256 message imprint submitted to timestamp authority; signed TimeStampResp returned.

Disputing requires alleging miner collusion or TSA compromise---specific allegations against third parties.

\subsection{Layer 3: Dual-Source Keystroke Validation}

Software injection (CGEventPost, SendInput) generates events in application-level streams but not hardware-level streams.

\textbf{Application Level:} Events through window server, including injected.

\textbf{Device Level:} Events from USB/Bluetooth HID, bypassing window server.

\begin{definition}[Validated Keystroke]
A keystroke is \emph{validated} if it appears in both streams within 50ms.
\end{definition}

The 50ms threshold reflects empirical USB HID polling intervals (1--8ms) plus OS scheduling jitter; legitimate keystrokes appear in both streams within 10ms under normal load.

Defeating dual-source validation requires kernel compromise ($C_3$).

\subsection{Layer 4: TPM/Secure Enclave}

Where available~\cite{tcg2011,coker2011}:
\begin{itemize}
    \item \textbf{Monotonic counter}: Prevents replay
    \item \textbf{Key sealing}: Secrets sealed to platform state
    \item \textbf{Attestation quotes}: Signed platform state
\end{itemize}

\subsection{Layer 5: Hash Chain Integrity}

Evidence stored in Merkle Mountain Range~\cite{crosby2009} (MMR):
$$\mathsf{leaf} = H(\mathtt{0x00} \| \mathsf{content\_hash} \| \mathsf{metadata\_hash})$$

Modifying any leaf breaks the chain; the adversary cannot update subsequent hashes without recomputing the entire tail.

\subsection{The Conjunction Barrier}

To dispute evidence, an adversary must allege:
\begin{enumerate}
    \item Jitter seal forged AND secret stolen in specific window
    \item VDF proofs invalid AND violation undetectable
    \item Hash chain modified AND hashes still verify
    \item Dual-source passed AND kernel compromised
    \item External anchors false AND Bitcoin/TSA colluded
\end{enumerate}

Each allegation names a different system, trust boundary, and time window.

\section{Security Analysis}
\label{sec:security}

\subsection{Residual Attacks}

\begin{table}[ht]
\centering
\small
\begin{tabular}{@{}lll@{}}
\toprule
\textbf{Attack} & \textbf{Prerequisites} & \textbf{Level} \\
\midrule
Environment compromise & Kernel before capture & $C_3$ \\
Clock manipulation & Admin + gap before anchor & $C_2$ \\
Key theft & Access within lifecycle & $C_2$--$C_3$ \\
VDF bypass & Faster hardware & $C_3$--$C_4$ \\
Anchor manipulation & Bitcoin 51\% or TSA & $C_4$ \\
Type AI content & Real-time typing & $C_1$ \\
\bottomrule
\end{tabular}
\caption{Residual attacks and required capabilities.}
\label{tab:attacks}
\end{table}

\textbf{Environment compromise.} A kernel-level adversary present before capture can falsify all observations. Disputing requires alleging a specific compromise mechanism---testable via forensic examination.

\textbf{Clock manipulation.} External anchors constrain windows. If Bitcoin confirms at $T_b$ and evidence claims $T_l < T_b$, manipulation must have occurred during $[T_l, T_b]$---bounded and specific.

\textbf{Typing AI content.} An attacker who types AI-generated text produces valid evidence. This is the economic security bound: forgery requires real-time typing.

\subsection{Tiered Security}

\begin{table}[ht]
\centering
\small
\begin{tabular}{@{}llll@{}}
\toprule
\textbf{Tier} & \textbf{Trust Assumption} & \textbf{Hardware} & \textbf{Defeats} \\
\midrule
1 & Machine trusted & None & $C_1$ \\
2 & OS untrusted & TPM/Enclave & $C_2$ \\
3 & Input untrusted & Attested keyboard & $C_3$ \\
4 & Nothing local & Network & All local \\
\bottomrule
\end{tabular}
\caption{Security tiers.}
\label{tab:tiers}
\end{table}

Full offline security (nothing local trusted) is impossible without hardware roots. With kernel access, an adversary can dump the session secret and forge evidence. Tier 2 is achievable with TPM 2.0 or Secure Enclave, present on $>$95\% of modern PCs.

\subsection{Worked Dispute Scenario}

A 47-minute document session. Challenger asserts ``this was fabricated.''

\begin{table}[ht]
\centering
\small
\begin{tabular}{@{}llll@{}}
\toprule
\textbf{Doubt} & \textbf{Forced Allegation} & \textbf{Level} & \textbf{Test} \\
\midrule
``Backdated'' & Clock set in window & $C_2$ & Anchor \\
``Forged'' & Secret stolen & $C_3$ & Forensics \\
``Chain fake'' & DB access + recompute & $C_3$ & Hash \\
``Anchors lie'' & TSA/BTC colluded & $C_4$ & Third-party \\
\bottomrule
\end{tabular}
\caption{Vague doubt becomes specific allegation.}
\label{tab:dispute}
\end{table}

\section{Evaluation}
\label{sec:eval}

\subsection{Jitter Seal Verification}

\begin{table}[ht]
\centering
\small
\begin{tabular}{@{}lll@{}}
\toprule
\textbf{Scenario} & \textbf{Trials} & \textbf{Outcome} \\
\midrule
Valid proof (baseline) & 1,000 & Verification succeeded \\
Fabricated jitter values & 10,000 & Verification failed \\
Mismatched document & 10,000 & Verification failed \\
Incorrect secret & 10,000 & Verification failed \\
\bottomrule
\end{tabular}
\caption{Jitter seal verification (31,000 trials).}
\label{tab:jitter-experiments}
\end{table}

All 1,000 valid proofs verified; all 30,000 attack trials failed. Forgery probability: $(1/2500)^{50} \approx 10^{-170}$.

\subsection{System Test Coverage}

\begin{table}[ht]
\centering
\small
\begin{tabular}{@{}lll@{}}
\toprule
\textbf{Component} & \textbf{Go Tests} & \textbf{Rust Tests} \\
\midrule
VDF computation/verification & 28 & 28 \\
MMR append/proof/verify & 16 & 12 \\
Checkpoint chain integrity & 18 & 16 \\
Evidence packet roundtrip & 24 & 31 \\
Key hierarchy \& ratcheting & 12 & 14 \\
TPM/Secure Enclave binding & 8 & 10 \\
Jitter seal verification & 24 & 18 \\
\midrule
\textbf{Total} & 320+ & 700+ \\
\bottomrule
\end{tabular}
\caption{Test coverage by component.}
\label{tab:tests}
\end{table}

\subsection{Performance}

\begin{table}[ht]
\centering
\small
\begin{tabular}{@{}lll@{}}
\toprule
\textbf{Operation} & \textbf{Time} & \textbf{Notes} \\
\midrule
Jitter computation & 286 ns & Per sample \\
Jitter injection & 500--3000\,$\mu$s & Intentional delay \\
Document hash (10 KB) & $\sim$85\,$\mu$s & SHA-256 \\
VDF checkpoint & $\sim$50 ms & Background \\
MMR append & $<$ 1 ms & Negligible \\
\bottomrule
\end{tabular}
\caption{Performance characteristics.}
\label{tab:perf}
\end{table}

The jitter range (0.5--3.0\,ms) is 5--27$\times$ below perception thresholds~\cite{card1983}.

\section{Limitations}
\label{sec:limitations}

\textbf{Content-agnostic.} The system attests to process, not cognitive origin. If an author types AI output, evidence shows genuine human typing.

\textbf{Requires cooperation.} Evidence cannot be generated retroactively.

\textbf{Secret management.} Secret compromise enables forgery; requires security comparable to private keys.

\textbf{Kernel-level adversaries.} A $C_3$ adversary present before capture defeats the system---but must make specific allegations to dispute.

\textbf{Labor markets.} Typing speed constraints bind individual humans; organized fraud can hire typists (linear cost scaling).

\section{Related Work}
\label{sec:related}

\textbf{Proof-of-X primitives.} Proof-of-work~\cite{nakamoto2008}, VDFs~\cite{boneh2018,pietrzak2019}, and proof-of-space~\cite{dziembowski2015} prove computational or storage resources. Proof-of-process requires physical interaction.

\textbf{Keystroke dynamics}~\cite{monrose2000,bergadano2002} focuses on authentication, not provenance. The jitter seal proves process rather than identifying individuals.

\textbf{Content watermarking}~\cite{kirchenbauer2023} embeds information in LLM output. The jitter seal watermarks the \emph{process}---timing, not text.

\textbf{Timestamping} \cite{haber1991,rfc3161,opentimestamps} proves existence at a time. Witnessd adds proof-of-process.

\textbf{Secure logging}~\cite{schneier1999} provides tamper evidence for sequential events. Witnessd adds external anchoring and proof-of-process.

\section{Conclusion}
\label{sec:conclusion}

We have presented the \emph{Adversarial Collapse Principle}: evidence succeeds when disputing it requires a conjunction of specific allegations against components with independent trust assumptions.

The jitter seal provides proof-of-process---cryptographic evidence that physical typing occurred. Witnessd combines jitter seals with VDFs, external anchors, dual-source validation, and TPM attestation. Each layer forces specific allegations at different capability levels.

The architecture does not prevent forgery. A sufficiently capable adversary can forge evidence. But ``it could have been faked'' is not an argument; ``the kernel was compromised via CVE-XXXX during the window from 14:00 to 14:47'' is.

This is the contribution: converting vague doubt into falsifiable allegations. In adversarial contexts, that difference matters.

\bibliographystyle{ACM-Reference-Format}

\begin{thebibliography}{99}

\bibitem{bellare1996}
M.~Bellare, R.~Canetti, and H.~Krawczyk.
\newblock Keying hash functions for message authentication.
\newblock In {\em CRYPTO}, pages 1--15, 1996.

\bibitem{bergadano2002}
F.~Bergadano, D.~Gunetti, and C.~Picardi.
\newblock User authentication through keystroke dynamics.
\newblock {\em ACM TISSEC}, 5(4):367--397, 2002.

\bibitem{boneh2018}
D.~Boneh, J.~Bonneau, B.~B{\"u}nz, and B.~Fisch.
\newblock Verifiable delay functions.
\newblock In {\em CRYPTO}, pages 757--788, 2018.

\bibitem{card1983}
S.~K. Card, T.~P. Moran, and A.~Newell.
\newblock {\em The Psychology of Human-Computer Interaction}.
\newblock Lawrence Erlbaum Associates, 1983.

\bibitem{coker2011}
G.~Coker et al.
\newblock Principles of remote attestation.
\newblock {\em Int. J. Information Security}, 10(2):63--81, 2011.

\bibitem{crosby2009}
S.~A. Crosby and D.~S. Wallach.
\newblock Efficient data structures for tamper-evident logging.
\newblock In {\em USENIX Security}, pages 317--334, 2009.

\bibitem{dziembowski2015}
S.~Dziembowski, S.~Faust, V.~Kolmogorov, and K.~Pietrzak.
\newblock Proofs of space.
\newblock In {\em CRYPTO}, pages 585--605, 2015.

\bibitem{haber1991}
S.~Haber and W.~S. Stornetta.
\newblock How to time-stamp a digital document.
\newblock {\em Journal of Cryptology}, 3(2):99--111, 1991.

\bibitem{kirchenbauer2023}
J.~Kirchenbauer et al.
\newblock A watermark for large language models.
\newblock In {\em ICML}, 2023.

\bibitem{merkle1980}
R.~C. Merkle.
\newblock Protocols for public key cryptosystems.
\newblock In {\em IEEE S\&P}, pages 122--134, 1980.

\bibitem{monrose2000}
F.~Monrose and A.~D. Rubin.
\newblock Keystroke dynamics as a biometric for authentication.
\newblock {\em Future Generation Computer Systems}, 16(4):351--359, 2000.

\bibitem{nakamoto2008}
S.~Nakamoto.
\newblock Bitcoin: A peer-to-peer electronic cash system.
\newblock 2008.

\bibitem{opentimestamps}
P.~Todd.
\newblock OpenTimestamps.
\newblock \url{https://opentimestamps.org}, 2016.

\bibitem{pietrzak2019}
K.~Pietrzak.
\newblock Simple verifiable delay functions.
\newblock In {\em ITCS}, pages 60:1--60:15, 2019.

\bibitem{rfc3161}
C.~Adams et al.
\newblock Time-Stamp Protocol (TSP).
\newblock RFC 3161, 2001.

\bibitem{schneier1999}
B.~Schneier and J.~Kelsey.
\newblock Secure audit logs to support computer forensics.
\newblock {\em ACM TISSEC}, 2(2):159--176, 1999.

\bibitem{tcg2011}
{Trusted Computing Group}.
\newblock TPM Main Specification Level 2, Version 1.2.
\newblock 2011.

\bibitem{wesolowski2019}
B.~Wesolowski.
\newblock Efficient verifiable delay functions.
\newblock In {\em EUROCRYPT}, pages 379--407, 2019.

\end{thebibliography}
\let\balance\relax

\end{document}